\documentclass[conference]{IEEEtran}
\usepackage{graphicx,indentfirst}
\usepackage{booktabs}
\usepackage{indentfirst}
\usepackage{amsmath}
\usepackage{cite}
\usepackage{stfloats}
\usepackage{multicol}
\usepackage{multirow}
\usepackage{slashbox}
\usepackage{subfig}
\usepackage[english]{babel}
\usepackage{amsthm}
\usepackage{bm}
\usepackage{amsmath}
\usepackage{multirow}
\usepackage{paralist}
\usepackage{algorithm, algorithmic}
\usepackage{amssymb}
\usepackage{amssymb}
\usepackage{mathrsfs}
\usepackage{amsfonts}
\usepackage{color}
\newtheorem{proposition}{Proposition}

\theoremstyle{plain}

\IEEEoverridecommandlockouts
\begin{document}
\bibliographystyle{IEEE2}
\title{Game Theoretic Analysis for Joint Sponsored and Edge Caching Content Service Market}
\author{Authors name}
\author{Zehui Xiong$^1$, Shaohan Feng$^1$, Dusit Niyato$^1$, Ping Wang$^2$, Amir Leshem$^3$ and Yang Zhang$^4$\\
$^1$School of Computer Science and Engineering, Nanyang Technological University (NTU), Singapore\\
$^2$Department of Electrical Engineering and Computer Science, York University, Canada\\
$^3$Faculty of Engineering, Bar-Ilan University, Israel\\
$^4$School of Computer Science and Technology, Wuhan University of Technology, China}
\maketitle

\begin{abstract}
With a sponsored content scheme in a wireless network, a sponsored content service provider can pay to a network operator on behalf of the mobile users/subscribers to lower down the network subscription fees at the reasonable cost in terms of receiving some amount of advertisements. As such, content providers, network operators and mobile users are all actively motivated to participate in the sponsored content ecosystem. Meanwhile, in 5G cellular networks, caching technique is employed to improve content service quality, which stores potentially popular contents on edge networks nodes to serve mobile users. In this work, we propose the joint sponsored and edge caching content service market model. We investigate an interplay between the sponsored content service provider and the edge caching content service provider under the non-cooperative game framework. Furthermore, a three-stage Stackelberg game is formulated to model the interactions among the network operator, content service provider, and mobile users. 
Sub-game perfect equilibrium in each stage is analyzed by backward induction. The existence of Stackelberg equilibrium is validated by employing the bilevel optimization programming. Based on the game properties, we propose a sub-gradient based iterative algorithm, which ensures to converge to the Stackelberg equilibrium.
\end{abstract}

\begin{IEEEkeywords}
Sponsored content, edge caching, content delivery network, Stackelberg game, bilevel optimization.
\end{IEEEkeywords}

\section{Introduction}\label{Sec:Introduction}
\textit{Sponsored content} in wireless networks has been introduced as a promising scheme to utilize a proper designed incentive mechanism to motivate mobile users to use more services and consume more contents, generating higher revenue to the mobile service ecosystem. For example, AT\&T launched a \textit{sponsored data plan} that allows content service providers to partly or fully pay for the data usages on behalf mobile users to the network operator, when the mobile users access the sponsored contents~\cite{AT&T}. With the sponsored content scheme, mobile users as the end users of network content services can offload their data subscription fees to content service providers. Correspondingly, content service providers can potentially cooperate with network providers by sponsoring data usages to motivate mobile users to join the network~\cite{xiong2017sponsor, xiong2018competition, wang2017hierarchical}.

In the meanwhile, \textit{edge caching} is emerging in the context of 5G networks as a promising technique to deliver high quality content services at lower costs. The key idea is that the edge caching content service provider caches the content on edge devices of the network in advance. Mobile users who request for contents may access the contents directly at the edge devices, without costly remote content deliveries. With edge caching, mobile users may obtain content services via local connections, e.g., WiFi, such that no cellular data usage will be incurred. Therefore, edge caching can be viewed as an alternative sponsorship scheme for mobile users. Edge caching is seen as on of the most effective solutions to deal with highly explosive content traffic demands in 5G era. 
As a result, equipping caches at small base stations offers a promising way to exploit the potential of edge networks in addition to densifying the existing cellular networks.

In this work, we introduce the integration of the edge caching and sponsored content schemes. In particular, we investigate their interaction and their coexistence that affect the mobile user behavior, the network operator in the data/content traffic market. To sum up, the main contributions of this paper are summarized as follows:
	\begin{enumerate}
	\item We formulate a joint sponsored and edge caching content service market model to analyze the interactions among the wireless network operator, the sponsored content service provider as well as the edge caching content service provider, and mobile users.

	\item We formulate a novel three-stage Stackelberg game to model their interactions to jointly maximize the payoff of the wireless network operator, the profit of each content service providers, and the utilities of mobile users.

	\item Through backward induction, we analyze the sub-game perfect equilibrium in each stage analytically. Furthermore, the existence of the Stackelberg equilibrium is validated by capitalizing on the bilevel optimization technique.
	\end{enumerate}

The rest of the paper is organized as follows. Section~\ref{Sec:Description} presents the system description and Section~\ref{Sec:Game} formulates a three-stage Stackelberg game to model the joint sponsored and edge caching content service market. The equilibrium analysis are shown in Section~\ref{Sec:Solution}. Section~\ref{Sec:Simulation} provides the performance evaluation. Section~\ref{Sec:Conclusion} concludes the paper.

\section{System Model} \label{Sec:Description}
We consider a market-oriented content delivery network consisting of four types of entities: a Wireless Network Operator (WNO), a Sponsored Content Service Provider (SCSP), an Edge Caching Content Service Provider (ECCSP), and a pool of Mobile Users (MUs). The MUs can access contents (e.g., video) from the SCSP and ECCSP. If an MU accesses contents from the SCSP, the contents can be downloaded directly through the network infrastructure provided by the WNO. In this case, the SCSP can sponsor for the data transfer from the WNO to the MU. In the mean time, the MU can also access contents stored in an edge caching device belonged to the ECCSP. The MU can download the contents from the device locally without involving the WNO. With the coexistence of both SCSP and ECCSP, the MU who requests for contents may autonomously choose the content source to lower down the content cost. The SCSP and ECCSP need to compete with each other to attract content demands from the users.

Let $y$ denote the content (volume) demanded by MUs, and ${\sigma_e}f(y)$ denote the utility obtained from accessing and consuming the content, where the factor ${\sigma _e}>0$ represents the utility coefficient of MUs, e.g., a particular valuation between MUs and content. Similar to that in~\cite{joe2015sponsoring}, we first define the following function:
	\begin{equation}\label{Eq:fx}
	f(y) = \frac{1}{{1 - \alpha }}{y^{1 - \alpha }},
	\end{equation}
where $0< \alpha <1$ is a given coefficient. In particular, the function $f(\cdot)$ is non-decreasing and concave with decreasing marginal satisfaction, indicating the decreasing marginal preference of MUs to contents. In traditional wireless content access, the WNO charges each MU a price $p$ for a unit volume of content downloaded. Thus, the general form of utility function of the MU with $y$ content demand is given defined as $v(y) = {\sigma _e} f(y) - py$.

With a sponsored content scheme, the payment from MUs to the WNO can be partly sponsored by the SCSP, as introduced in Section~\ref{Sec:Introduction}. Denote $\theta \in [0,1]$ as a sponsorship factor of content sponsored by the SCSP, i.e., $\theta$ units of the content is sponsored. Thus, the MU pays for the rest $(1-\theta)y$ units of content to the WNO, with the cost $(1 - \theta )py$ incurred to the MU~\cite{xiong2017sponsor, xiong2018competition}. Generally, the MUs are also affected by another variable $l_{a}$ which is the amount of advertisement imposed by the SCSP and ECCSP per volume of content. We assume that $l_a$ is constant for all contents. For example, Pandora Internet Radio plays advertisement at regular intervals between songs. We assume a normalized $l_a \in [0, 1]$ since both the ECCSP and SCSP have the amount of advertisements strictly less than that of the provided contents. For the ease of derivation later, we introduce an auxiliary variable $\tau$ defined as $\tau = \frac{1}{{1 + l_a}}, \tau \in\left[\frac{1}{2},1\right]$. Thus, the utility of the MU which has the content demand $y$ from the SCSP is expressed as $u_s(y) = \tau{\sigma _e} f(y) - (1 - \theta )py$.

An ECCSP is able to cache contents in the edge caching devices for MUs to access via a local network connection. We denote $t\in [0,1]$ as the caching effort of the ECCSP, which indicates the sponsorship from the ECCSP to MUs. When an MU access the cached content from the ECCSP, the advertisement from the ECCSP is imposed to the MU which lowers its utility. Nevertheless, accessing the cached content, the MU does not need to pay the WNO~\cite{pang2016joint}. Accordingly, the utility of the MU which has demand $z$ for the cached content from the ECCSP is expressed as ${u_e}(z) = \tau {\sigma _e}f(z)g(t) - cz$, where $c$ is the network handover cost. The caching effort $t$ indicates the amount of resources (e.g., storage and bandwidth) allocated for delivering the cached content to the MUs. The function $g(t)$ is defined as content delivery quality. In particular, $g(t)$ is a monotonically increasing function with diminishing returns reflecting the positive influence of caching effort $t$ on MUs' experience. Similar to~(\ref{Eq:fx}), we adopt the a common function to capture the Quality of Service (QoS) experienced by MUs, i.e.,
	\begin{equation}\label{Eq:gt}
		g(t) = \frac{1}{{1 - \beta }}{t^{1 - \beta }},
	\end{equation}
where $0 < \beta < 1$ is defined to be a coefficient.

\section{Game Formulation for Joint Sponsored and Edge Caching Content Service Model}\label{Sec:Game}

In this section, we formulate a three-stage Stackelberg game to model the interactions among the network participants (i.e., WNO, SCSP, ECCSP, and MUs). In the following Sections~\ref{Subsec:A},~\ref{Subsec:B} and~\ref{Subsec:C}, we analyze the sub-game problems of each stage in the game.

	\subsection{Stage III: Mobile Users (Followers)}\label{Subsec:A}
	In the network, each MU may either choose to access the sponsored contents from the SCSP, or the cached contents provided by the ECCSP. The MU determines a fraction of the content demand to access from the SCSP denoted by $x \in [0,1]$. The fraction of cached content demand from the ECCSP is thus $1-x$.

	Each rational MU balances the content access between the SCSP and ECCSP to maximize its utility. The utility of the MU can be expressed as follows:
		\begin{equation}\label{Eq:utility_model2}
			{u}(x;\theta,t,p) =\tau{\sigma _e} f(x) - (1 - \theta )xp + \tau{\sigma _e} f(1 - x)g(t) -(1 - x)c,
		\end{equation}
	where $f(x)$ and $g(t)$ are defined in~(\ref{Eq:fx}) and~(\ref{Eq:gt}), respectively.

	Each MU sub-game problem can be formulated as follows, given the volume unit price asked by the WNO, the sponsorship factor $\theta$ from the SCSP and the caching effort $t$ from the ECCSP, the MU chooses $x$ to maximize the utility:\\
	{\bf{Problem 1. (The MU sub-game):}}
		\begin{equation}
		\begin{aligned}
		& \underset{x}{\text{maximize}}
		& & u(x;\theta,t,p) \\
		& \text{subject to}
		& & x \in [0 ,1].\\
		\end{aligned}
		\end{equation}

	\subsection{Stage II: SCSP and ECCSP (2nd-Tier Players)}\label{Subsec:B}
	At this stage, the SCSP and ECCSP both determine their individual strategy simultaneously in a competitive manner, given the pricing strategy of the WNO.

		\subsubsection{Sponsored Content Service Provider}
		The SCSP aims to maximize the profit sponsoring, i.e., the gained advertisement revenue minus the cost of sponsorship provided for the MUs, as follows:
			\begin{equation}\label{Eq:profit_ssp}
				\Pi_s(\theta;p) = {\sigma _c} h(x) - \theta px.
			\end{equation}
		where $\sigma _c$ is the advertisement revenue coefficient, and the term ${\sigma _c}h(x)$ denotes the advertisement revenue~\cite{joe2015sponsoring}, in which $h(x)$ is defined as follows:
			\begin{equation}\label{Eq:hx}
				h(x) = \frac{1}{{1 - \gamma }}{x^{1 - \gamma }}	,
			\end{equation}
		where $0 < \gamma < 1$ which is a coefficient. Thus, the SCSP sub-game is defined as follows:\\
		{\bf{Problem 2-A. (The SCSP sub-game):}}
			\begin{equation}
			\begin{aligned}
			& \underset{x}{\text{maximize}}
			& & \Pi_s(\theta;p) \\
			& \text{subject to}
			& & \theta \in [0 ,1].\\
			\end{aligned}
			\end{equation}

		\subsubsection{Edge caching content service provider}\label{Subsec:C}
		As discussed in Section~\ref{Sec:Description}, the cost incurred to the ECCSP for caching the content with caching effort $t$ is denoted as $C t$. The profit gained by the ECCSP, i.e., the advertisement revenue from content traffic minus the cost for content caching, is defined as follows:
			\begin{equation}\label{Eq:profit_esp}
			  \Pi_e(t;p) = {\sigma _c}h(1-x) - Ct,
			\end{equation}
		where $h(\cdot)$ is the same as defined in~(\ref{Eq:hx}). We assume $C$ to be a baseline content caching cost. The cost of the ECCSP for caching the content with caching effort $t$ is thus $C t$~\cite{de2017competitive, tamoor2016caching, pang2016joint}. Accordingly, the profit maximization problem of the ECCSP is formulated as follows:\\
		{\bf{Problem 2-B. (The ECCSP sub-game):}}
			\begin{equation}
			\begin{aligned}
			& \underset{x}{\text{maximize}}
			& & \Pi_e(t;p) \\
			& \text{subject to}
			& & t \in [0 ,1].\\
			\end{aligned}
			\end{equation}

	\subsection{Stage I: Wireless Network Operator (1st-Tier Player)}
	The WNO sets the volume unit price $p$ for data traffic. In addition to obtaining its revenue from charging the SCSP and MUs, the WNO also has the content delivery cost. Accordingly, the objective of the WNO is to maximize its payoff, expressed as follows:
		\begin{equation}\label{Eq:payoff_WNO}
			{\mathscr P}(p) = px - w x^2,
		\end{equation}
	where $w x^2$ denotes the corresponding cost, and $w$ represents the unit cost of content delivery. The quadratic sum form reflects the marginal cost increases as the total demand increases, e.g., due to congestion effects, which is a well-accepted assumption. The strategy space of the WNO is denoted as $\{{\bf P}: 0\le p \le \overline p\}$, where $\overline p$ is the maximum price. As a result, the payoff maximization problem of the WNO is formulated as follows:\\
	{\bf{Problem 3. (The WNO sub-game):}}
		\begin{equation}
		\begin{aligned}
			& \underset{x}{\text{maximize}}
			& & {\mathscr P}(p) \\
			& \text{subject to}
			& & p \in [0,\overline p].\\
		\end{aligned}
		\end{equation}

	{\bf Problems 1}, {\bf 2-A}, {\bf 2-B} and {\bf 3} altogether form a three-stage Stackelberg game with the objective of finding the Stackelberg equilibrium. In the next section, the sub-game problems will be solved sequentially. The Stackelberg equilibrium will be investigated.

\section{Game Equilibrium Analysis}\label{Sec:Solution}
	\begin{figure*}
	\begin{equation}\label{Eq:SSPfirstorder}
	\frac{{\partial {x^*}}}{{\partial \theta }} = \frac{p}{{\alpha \tau {\sigma _e}\left[ {{{x^*}^{ - \alpha - 1}} + \frac{{{t^{1 - \beta }}}}{{1 - \beta }}{{(1 - {x^*})}^{ - \alpha - 1}}} \right]}}>0.
	\end{equation}
	\begin{equation}\label{Eq:SSPsecondorder}
	\frac{{{\partial ^2}{x^*}}}{{\partial {\theta ^2}}} = \left[ { - {{x^*}^{ - \alpha - 2}} + \frac{{{t^{1 - \beta }}}}{{1 - \beta }}{{(1 - {x^*})}^{ - \alpha - 2}}} \right]\frac{{ - (\alpha + 1){p^2}}}{{{\alpha ^2}{\tau ^2}{\sigma _e}^2{{\left[ {{{x^*}^{ - \alpha - 1}} + \frac{{{t^{1 - \beta }}}}{{1 - \beta }}{{(1 - {x^*})}^{ - \alpha - 1}}} \right]}^3}}}.
	\end{equation}
	\begin{equation}\label{Eq:ESPfirstorder}
	\frac{{\partial {x^*}}}{{\partial t}} = \frac{{ - {t^{ - \beta }}{{(1 - {x^*})}^{ - \alpha }}}}{{\alpha \left[ {{{x^*}^{ - \alpha - 1}} + \frac{{{t^{1 - \beta }}}}{{1 - \beta }}{{(1 - {x^*})}^{ - \alpha - 1}}} \right]}}<0.
	\end{equation}
	\begin{multline}\label{Eq:ESPsecondorder}
	\frac{{{\partial ^2}{x^*}}}{{\partial {t^2}}} = \frac{{{t^{ - \beta }}{{(1 - {x^*})}^{ - \alpha }}}}{{\alpha \Big[ {{{x^*}^{ - \alpha - 1}} + \frac{{{t^{1 - \beta }}}}{{1 - \beta }}{{(1 - {x^*})}^{ - \alpha - 1}}} \Big]}}\Bigg\{ \frac{{2{t^{ - \beta }}{{(1 - {x^*})}^{ - \alpha - 1}}}}{{{{x^*}^{ - \alpha - 1}} + \frac{{{t^{1 - \beta }}}}{{1 - \beta }}{{(1 - {x^*})}^{ - \alpha - 1}}}} + \frac{\beta}{t} \\+ \frac{{ - (\alpha + 1){t^{ - \beta }}{{(1 - {x^*})}^{ - \alpha }}\Big[ { - {{x^*}^{ - \alpha - 2}} + \frac{{{t^{1 - \beta }}}}{{1 - \beta }}{{(1 - {x^*})}^{ - \alpha - 2}}} \Big]}}{{\alpha {{\Big[ {{{x^*}^{ - \alpha - 1}} + \frac{{{t^{1 - \beta }}}}{{1 - \beta }}{{(1 - {x^*})}^{ - \alpha - 1}}} \Big]}^2}}}
	\Bigg\}.
	\end{multline}
	\end{figure*}
	
In this section, we sequentially solve the sub-game perfect equilibrium in each stage of Stackelberg game by employing backward induction. Specifically, we analyze the content demand of MUs in Stage~III, the sponsoring strategy of the SCSP as well as the caching strategy of the ECCSP in Stage~II, and the pricing strategy of the WNO in Stage~I, respectively.

	\subsection{Stage III: MU's content demand}\label{Subsec:StageIII}
	Given the unit content downloading price $p$ charged by the WNO, sponsorship factor $\theta$ determined by the SCSP, as well as caching effort $t$ determined by the ECCSP, respectively, the MUs determine optimal content demands for utility maximization individually in the sub-game $\mathcal{G}^u$. Thus, we analyze the sub-game $\mathcal{G}^u$ by solving the~{\bf Problem~1}. By substituting $f(\cdot)$ as in~(\ref{Eq:fx}) and $g(\cdot)$ as in~(\ref{Eq:gt}) into~(\ref{Eq:utility_model2}), the utility of the MU with the decision variable $x$ is obtained as follows:
		\begin{multline}\label{Eq:Utility}
		u(x;\theta ,t,p)= \frac{{\tau {\sigma _e}}}{{1 - \alpha }}{x^{1 - \alpha }} + \frac{{\tau {\sigma _e}{t^{1 - \beta }}}}{{(1 - \alpha )(1 - \beta )}}{(1 - x)^{1 - \alpha }} \\- (1 - x)c - (1 - \theta )xp.
		\end{multline}

	In the next step, the first order and second order derivatives of~(\ref{Eq:Utility}) with respect to $x$ is taken to prove its concavity, as follows:
		\begin{equation}\label{Eq:FirstOrderUtility}
		\frac{{\partial {u}}}{{\partial {x}}} = \tau {\sigma _e}{x^{ - \alpha }} - \frac{{\tau {\sigma _e}{t^{1 - \beta }}}}{{1 - \beta }}{(1 - x)^{ - \alpha }} + c - (1 - \theta )p,
		\end{equation}
		\begin{equation}\label{Eq:SecondOrderUtility}
		\frac{{{\partial ^2}{u}}}{{\partial {x}^2}} = - \alpha \tau {\sigma _e}{x^{ - \alpha - 1}} - \frac{{\alpha \tau {\sigma _e}{t^{1 - \beta }}}}{{1 - \beta }}{(1 - x)^{ - \alpha - 1}}< 0.
		\end{equation}

	As shown in (\ref{Eq:FirstOrderUtility}), the second order derivative of $u(\cdot)$ is always negative. Therefore, the function $u$ is strictly concave with respect to $x$. As the strategy space of $x$ is already known to be a convex and compact subset of the Euclidean space, the following proposition can be concluded accordingly~\cite{han2012game}.
		\begin{proposition}
			The sub-game perfect equilibrium in the sub-game $\mathcal{G}^u$ is unique.
		\end{proposition}
	Furthermore, based on the first order derivative condition, we have
		\begin{equation}\label{Eq:FirstOrderDerivative}
		\frac{{\partial {u}}}{{\partial {x}}} = \tau {\sigma _e}{x^{ - \alpha }} - \frac{{\tau {\sigma _e}{t^{1 - \beta }}}}{{1 - \beta }}{(1 - x)^{ - \alpha }} + c - (1 - \theta )p = 0,
		\end{equation}
	and we can show that
		\begin{equation}\label{Eq:UtilityEquilibrium}
		{x^{ - \alpha }} - \frac{{{t^{1 - \beta }}}}{{1 - \beta }}{(1 - x)^{ - \alpha }} = \frac{1}{{\tau {\sigma _e}}}\bigg[ {\left( {1 - \theta } \right)p - c} \bigg].
		\end{equation}

	The following conclusions hold. $x^* = 1$ given the best response of the MU is larger than $1$; $x^* = 0$ given the best response of the MU is less than $0$. If the best response of the MU with respect to the content demand $x^*$ is within the strategy space $[0, 1]$, the best response of the MU, i.e., $x^*$, satisfies the condition in~(\ref{Eq:UtilityEquilibrium}), i.e.,
		\begin{equation}\label{Eq:UtilityEquilibriumForm}
		{{x^*}^{ - \alpha }} - \frac{{{t^{1 - \beta }}}}{{1 - \beta }}{(1 -{x^*})^{ - \alpha }} = \frac{1}{{\tau {\sigma _e}}}\left[ {\left( {1 - \theta } \right)p - c} \right].
		\end{equation}

	\subsection{Stage II: SCSP's sponsoring and ECCSP's caching strategy}\label{Subsec:StageII}
	Given the sub-game equilibrium in $\mathcal{G}^u$ obtained from the Stage~III, the SCSP and ECCSP as 2nd-tier players optimize the sponsoring and caching strategies for profit maximization in a competitive manner, respectively. The optimal strategies of both the SCSP and ECCSP are obtained by solving the sub-game {\bf Problems 2-A} and {\bf 2-B}.

	Firstly, we analyze the optimal sponsoring strategy of the SCSP. Based on (\ref{Eq:profit_ssp}), the profit obtained by the SCSP can be reformulated as follows:
		\begin{equation}\label{Eq:sspprofit}
		\Pi_s(\theta;p) = {\sigma _c}\frac{1}{{1 - \gamma }}{{x^*}^{1 - \gamma }} - \theta px^*,
		\end{equation}
	where $x^*$ is the best response of the MU. The first and second order derivatives of profit ${\Pi _s}(\theta ;p)$ with respect to the sponsorship factor $\theta$ are given as follows:
		\begin{equation}\label{Eq:FirstOrderProfit_s}
		\frac{{\partial {\Pi _s}(\theta ;p)}}{{\partial \theta }} = {\sigma _c}{{x^*}^{ - \gamma }}\frac{{\partial {x^*}}}{{\partial \theta }} - p{x^*} - \theta p\frac{{\partial {x^*}}}{{\partial \theta }},
		\end{equation}
	and
		\begin{multline}\label{Eq:SecondOrderProfit_s}
		\frac{{{\partial ^2}{\Pi _s}(\theta ;p)}}{{\partial {\theta ^2}}} = - \gamma {\sigma _c}{{x^*}^{ - \gamma - 1}}{\left( {\frac{{\partial {x^*}}}{{\partial \theta }}} \right)^2} + {\sigma _c}{{x^*}^{ - \gamma }}\frac{{{\partial ^2}{x^*}}}{{\partial {\theta ^2}}} \\- 2p\frac{{\partial {x^*}}}{{\partial \theta }} - \theta p\frac{{{\partial ^2}{x^*}}}{{\partial {\theta ^2}}}.
		\end{multline}

	From the condition in~(\ref{Eq:UtilityEquilibrium}), we obtain $\frac{{\partial {x^*}}}{{\partial \theta }}$ and $\frac{{{\partial ^2}{x^*}}}{{\partial {\theta ^2}}}$ with simple steps, as shown in~(\ref{Eq:SSPfirstorder}) and~(\ref{Eq:SSPsecondorder}).

	Likewise, we analyze the optimal caching strategy of the ECCSP. From~(\ref{Eq:profit_esp}), the profit obtained by the ECCSP can be expressed as follows:
		\begin{equation}\label{Eq:espprofit}
		{\Pi_e}(t;p) = {\sigma _c}\frac{1}{{1 - \gamma }}{(1 - {x^*})^{1 - \gamma }} - Ct.
		\end{equation}
	The first and second derivatives of profit $\Pi_e(t;p)$ with respect to caching effort $t$ are given as follows:
		\begin{equation}\label{Eq:FirstOrderProfit_c}
		\frac{{\partial {\Pi_e}(t;p)}}{{\partial t}} = {\sigma _c}{(1 - {x^*})^{ - \gamma }}\frac{{\partial {x^*}}}{{\partial t}} - C,
		\end{equation}
	and
		\begin{multline}\label{Eq:SecondOrderProfit_c}
		\frac{{{\partial ^2}{\Pi_e}(t;p)}}{{\partial {t^2}}} = - \gamma {\sigma _c}{(1 - {x^*})^{ - \gamma - 1}}{\left( {\frac{{\partial {x^*}}}{{\partial t}}} \right)^2} \\- {\sigma _c}{(1 - {x^*})^{ - \gamma }}\frac{{{\partial ^2}{x^*}}}{{\partial {t^2}}}.
		\end{multline}

	From the condition in~(\ref{Eq:UtilityEquilibrium}), we derive $\frac{{\partial {x^*}}}{{\partial t}}$ and $\frac{{{\partial ^2}{x^*}}}{{\partial {t^2}}}$ with simple steps, as shown in~(\ref{Eq:ESPfirstorder}) and~(\ref{Eq:ESPsecondorder}). By analyzing the profits of the SCSP and ECCSP given in~(\ref{Eq:sspprofit}) and~(\ref{Eq:espprofit}), respectively, we have the following proposition.
		\begin{proposition}
		The existence of the Nash equilibrium in the non-cooperative sub-game $\mathcal{G}^{c}$ is guaranteed if the following conditions
		\begin{equation}\label{Eq:Condition1}
		{\sigma _c}{{x^*}^{ - \gamma }} - \theta p > 0,
		\end{equation}
		\begin{equation}\label{Eq:Condition2}
		- {{x^*}^{ - \alpha - 2}} + \frac{{{t^{1 - \beta }}}}{{1 - \beta }}{(1 - {x^*})^{ - \alpha - 2}}>0,
		\end{equation}
		and
		\begin{equation}\label{Eq:Condition3}
		\gamma + \alpha - 1 > 0
		\end{equation}
		hold.
		\end{proposition}

		\begin{proof}
		From~(\ref{Eq:SSPsecondorder}), we can easily know that $\frac{{{\partial ^2}{x^*}}}{{\partial {\theta ^2}}}<0$ under the condition in~(\ref{Eq:Condition2}). Further, we have $\frac{{\partial {x^*}}}{{\partial \theta }} >0$ and $
		\frac{{{\partial ^2}{\Pi _s}(\theta ;p)}}{{\partial {\theta ^2}}} = - \gamma {\sigma _c}{{x^*}^{ - \gamma - 1}}{\left( {\frac{{\partial {x^*}}}{{\partial \theta }}} \right)^2} - (\sigma _c {{x^*}^{ - \gamma }}-\theta p)\frac{{{\partial ^2}{x^*}}}{{\partial {\theta ^2}}} - 2p\frac{{\partial {x^*}}}{{\partial \theta }}$.
		Consequently, we can conclude that $\frac{{{\partial ^2}{\Pi _s}(\theta ;p)}}{{\partial {\theta ^2}}} < 0$ under the condition in~(\ref{Eq:Condition1}). Then, we analyze the properties of $\frac{{{\partial ^2}{\Pi_e}(t;p)}}{{\partial {t^2}}}$. From~(\ref{Eq:SecondOrderProfit_c}), in order to prove that $\frac{{{\partial ^2}{\Pi_e}(t;p)}}{{\partial {t^2}}}$ is negative, we need to prove that
		\begin{equation}\label{Eq:step1}
		\gamma {(1 - {x^*})^{ - 1}}{\left( {\frac{{\partial {x^*}}}{{\partial t}}} \right)^2} + \frac{{{\partial ^2}{x^*}}}{{\partial {t^2}}}>0.
		\end{equation}
		By substituting~(\ref{Eq:ESPfirstorder}) and~(\ref{Eq:ESPsecondorder}) into~(\ref{Eq:step1}), and with simple steps, we have
		\begin{multline}\label{Eq:step2}
		\underbrace {\left( {\gamma + 2\alpha } \right){{(1 - {x^*})}^{ - 1}}{{x^*}^{ - \alpha - 1}}}_{ > 0} + \underbrace {(\alpha + 1){{x^*}^{ - \alpha - 2}}}_{ > 0} \\+ (\gamma + \alpha - 1)\underbrace {\frac{{{t^{1 - \beta }}}}{{1 - \beta }}{{(1 - {x^*})}^{ - \alpha - 2}}}_{ > 0} >0.
		\end{multline}
		Accordingly, under the condition in~(\ref{Eq:Condition3}), the inequality given in~(\ref{Eq:step2}) is satisfied. Therefore, the negativity of $\frac{{{\partial ^2}{\Pi_e}(t;p)}}{{\partial {t^2}}}$ is proven.
		Accordingly, the non-cooperative sub-game $\mathcal{G}^{c}$ between the SCSP and ECCSP is a concave game, from which the existence of the Nash equilibrium in the sub-game $\mathcal{G}^{c}$ follows.
		\end{proof}

	Finally, we prove the uniqueness of the Nash equilibrium in the non-cooperative sub-game between the SCSP and ECCSP, as shown in the following proposition.
		\begin{proposition}	\label{prop.2}
		If there exists at least one Nash equilibrium in the non-cooperative sub-game $\mathcal{G}^{c}$, the Nash equilibrium is unique provided that the following condition
		\begin{equation}\label{Eq:Condition4}
		2\alpha - 1 > 0
		\end{equation}
		holds.
		\end{proposition}
		\begin{proof}[Proof Sketch]
		We mainly prove that the Jacobian matrix of point-to-set mapping with respect to the profit profile of both the SCSP and ECCSP is negative definite. Consequently, $\nabla {\bf{F}}$ is diagonally strictly concave, from which the uniqueness of the Nash equilibrium in the non-cooperative sub-game $\mathcal{G}^{c}$ is guaranteed~\cite{han2012game}. Please refer to~\cite{xiong2017jsac} for more details.
		\end{proof}

	\subsection{Stage I: WNO's pricing strategy}\label{Subsec:StageI}
	In Stage~I, the monopolist WNO determines the optimal price $p^*$ by solving the {\bf Problem 3} with the optimal sponsorship factor $\theta^*$ as well as the optimal caching effort $t^*$ obtained at Stage~II, and the optimal content demand $x^*$ obtained at Stage~I. In this case, the {\bf Problem~3} can be reformulated as follows:
		\begin{align}\label{Eq:NewProblem3}
		 \text{maximize} \quad &\null {\mathscr P}(p) = px^* - w {x^*}^2 \\
		 \text{subject to} \quad &\null 0\le p\le \overline p, \label{Eq:1constraint}\\
		 & {\theta ^*} = \arg \max_\theta {\Pi _s}, \label{Eq:2constraint}\\
		 & {t^*} = \arg \max_t {\Pi_e}, \label{Eq:3constraint}\\
		 & {\sigma _c}{{x^*}^{ - \gamma }} - \theta^* p > 0,\label{Eq:4constraint}\\
		 & - {{x^*}^{ - \alpha - 2}} + \frac{{{{t^*}^{1 - \beta }}}}{{1 - \beta }}{(1 - {x^*})^{ - \alpha - 2}}>0, \label{Eq:5constraint}\\
		 & {{x^*}^{ - \alpha }} - \frac{{{{t^*}^{1 - \beta }}}}{{1 - \beta }}{(1 - x^*)^{ - \alpha }} = \frac{1}{{\tau {\sigma _e}}}\left[ {\left( {1 - \theta^* } \right)p - c} \right] \label{Eq:6constraint}.
		\end{align}
	The constraint in~(\ref{Eq:1constraint}) applies for the strategy space of the WNO. The constraints in~(\ref{Eq:2constraint}) and~(\ref{Eq:3constraint}) indicate that $\theta^*$ and $t^*$ denote the best responses of the SCSP and ECCSP, respectively, given the price $p$. The constraint in~(\ref{Eq:4constraint}) is derived from the best response of the MU given $\theta$, $t$ and $p$, which represents the implicit function of $x^*(\theta, t, p)$. The constraints in~(\ref{Eq:5constraint}) and~(\ref{Eq:6constraint}) ensure the existence and the uniqueness of the non-cooperative sub-game~${\cal G}_c$, which are given in Proposition~\ref{prop.2}.

	Let $\cal X$ denote $x^*(\theta, t, p): {\cal D}_\theta {\otimes} {\cal D}_t {\otimes} {\cal D}_p \mapsto {\cal D}_x$, where ${\otimes}$ denotes the Cartesian product, ${\cal D}_\theta$, ${\cal D}_t$ and ${\cal D}_p$ represent the domains of $\theta$, $t$ and $p$, respectively. Note that these domains are all close sets. Further, we define $\Pi (\theta ,t,p)$ and $g(\theta ,t,p)$ as follows: $\Pi (\theta ,t,p) = \left[ {\begin{array}{*{20}{c}}
	{{\sigma _c}h({\cal X}) - \theta p{\cal X}}\\
	{{\sigma _c}h(1 - {\cal X}) - Ct}
	\end{array}} \right]$, $g(\theta ,t,p) =\left[ \begin{array}{*{20}{c}}
	{ - {\sigma_c}{{\cal X}^{ - \alpha }} + \theta p + \upsilon }\\
	{{\cal X}^{ - \alpha - 2}} - \frac{{{t^{1 - \beta }}}}{{1 - \beta }}{{(1 - {\cal X})}^{ - \alpha - 2}} + \upsilon \end{array} \right]$, where $\upsilon$ is a small number. In addition, we define ${\bm \rho}=[\theta; t]: {\cal D}_\theta {\otimes} {\cal D}_t \mapsto {\cal D}_{\bm \rho}$. As a result, {\bf Problem~3} can be redefined as the bilevel programming problem, which is shown as follows:
		\begin{equation}\label{Eq:NewProblem4}
		\begin{aligned}
		& \underset{0<p<\overline p}{\text{maximize}}
		& & {\mathscr P}(p) = p {\cal X} - w {\cal X}^2\\
		& \text{subject to}
		& & {{\bm {\bm \rho}^*}} = \arg \max_{\bm \rho} {\Pi({\bm \rho},p)},\\
		& & &{g({\bm \rho},p)} \le 0,\\
		& & & {\bm \rho} \in {\cal D}_{\bm \rho}.
		\end{aligned}
		\end{equation}

	As shown in Section~\ref{Subsec:StageII}, we have proved the existence of a unique pair of $\theta^*$ and $t^*$ for any given $p$ as the optimal solution of {\bf Problem~2}, i.e., the Nash equilibrium in non-cooperative sub-game~${\cal G}_c$. Accordingly, the optimal solution to the lower-level programming problem $\bm{{\bm {\bm \rho}^*}} = [{\theta ^*};{t^*}]$ exists and is unique, for any given $p$. Therefore, the strong sufficient optimality condition of second order (SSOSC) is satisfied for the bilevel programming problem in~(\ref{Eq:NewProblem4}) because of the existence and the uniqueness of ${\bm {\bm \rho}^*}$ (Theorem~3.9 in~\cite{dempe2015bilevel}). This indicates that the optimal solution to the lower-level programming problem of the bilevel programming problem is strongly stable. Thus, our bilevel programming problem can be reduced to a single-level problem, which is expressed as follows:
		\begin{equation}\label{Eq:NewProblem5}
		\begin{footnotesize}
		\begin{aligned}
		& \underset{0<p<\overline p}{\text{maximize}}
		& & {\mathscr P}(p) = p {\cal X} - w {\cal X}^2,\\
		& \text{subject to}
		& & {{\bm \rho}} = U(p).\\
		\end{aligned}
		\end{footnotesize}
		\end{equation}
		${{\bm \rho}} =U(p)$ in the constraint can be obtained by using the KKT condition to the lower-level programming problem of the bilevel programming problem as follows:
		\begin{equation}
		\left\{ {\begin{array}{*{20}{c}}
		{{\nabla _{\bm \rho} }\Pi ({\bm \rho} ,p) - {\lambda ^\top}{\nabla _{\bm \rho} }g({\bm \rho} ,p) = 0},\\
		{0 \le \lambda \bot - g({\bm \rho} ,p) \ge 0},
		\end{array}} \right.
		\end{equation}
	where $\lambda$ denotes the Lagrangian multiplier vector. Moreover, the feasible domain of the single-level programming is defined as $\Omega (p,{\bm \rho} ) = \left\{ {\left. {(p,{\bm \rho} )} \right|{\bm \rho} = U(p),p \in {{\cal D}_p}} \right\}$, which is a non-empty and closed set according to the Weierstrass Theorem~\cite{sundaram1996first}. Since we know that the optimal solution to the lower-level programming problem of the bilevel programming problem is unique, Constant Rank Constraint Qualification (CRCQ) and Mangasarian-Fromovitz Constraint Qualification (MFCQ) are satisfied by all the feasible points in $\Omega (p,{\bm \rho} )$ (Theorem 3.9 in~\cite{dempe2015bilevel}). Based on Theorem 4.10 in~\cite{dempe2002foundations}, ${\bm \rho}=U(p)$ is a piecewise continuously differentiable function and $(p,{\bm \rho})=(p, U(p))$ is therefore continuous on $p$. Further, with the closed sets $\Omega (p,{\bm \rho} )$ and ${\cal D}_p$, according to the well-known Closed Graph Theorem, we can conclude that ${\bm \rho} = U(p)$ being continuous implies that the mapping ${\cal D}_p \mapsto \Omega (p,{\bm \rho} )$ is closed. Therefore, $\Omega (p,{\bm \rho} )$ is non-empty and closed, and thus the bilevel programming problem admits a globally optimal solution, i.e., the Stackelberg equilibrium. Accordingly, we conclude with the following proposition.
		\begin{proposition}
		There exists at least one Stackelberg equilibrium in the proposed three-stage Stackelberg game.
		\end{proposition}

	Each stage of the proposed three-stage Stackelberg game has been investigated so far. Similar to that in~\cite{zhang2017hierarchical}, we then present the sub-gradient based algorithm to obtain the Stackelberg equilibrium of the proposed game. The convergence property of the sub-gradient based iterative algorithm has been proved in~\cite{zhang2017hierarchical}.

\section{Performance Evaluation}\label{Sec:Simulation}
We employ simulations to evaluate the network participant performance metrics in the proposed joint sponsored and edge caching content service market, with default network parameters set as follows:  $\alpha = 0.8$, $\beta = 0.5$, $\gamma = 0.8$, $l_a = 1$, $\sigma_e = 40$, $\sigma_c = 120$, $c = 80$, $C = 120$, $w = 1$ and $\overline p =100$.

Firstly, the impact of the price constraint on the WNO is investigated, as in Fig.~\ref{Fig:Performance_price_bar}. It is worth noting that the optimal price offered by the WNO is the same as the maximum price constraint. The intuition is that the lower price can attract the MU to consume more sponsored content from the SCSP, which may incur higher delivery cost maintained by the WNO. Thus, the WNO is reluctant to lower the offered price. In addition, we find that as the price constraint, i.e., the optimal price increases, the payoff of the WNO increases and the sponsored content demand of the MU decreases. This is because the WNO is able to set a higher price and extract more surplus from the MU, and achieves higher payoff consequently. Further, as expected, the payoff of the WNO decreases with the increase of $w$, i.e., the unit cost of content delivery.

	\begin{figure}[t]
	\centering
	\includegraphics[width=.45\textwidth]{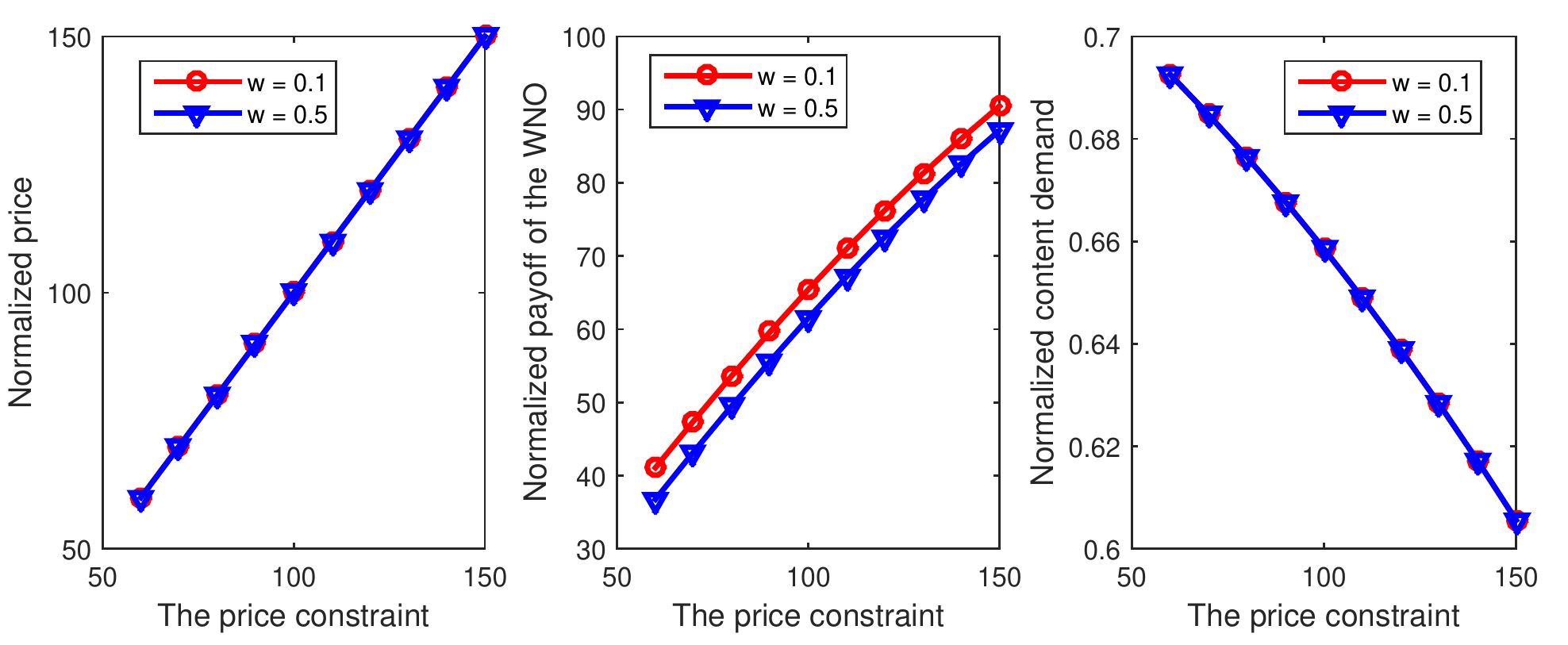}
	\caption{\small{The impact of the price constraint $\overline p$ on the WNO from joint sponsored and edge caching content service market.}}\label{Fig:Performance_price_bar}
	\end{figure}

	\begin{figure}[t]
	\centering
	\includegraphics[width=.45\textwidth]{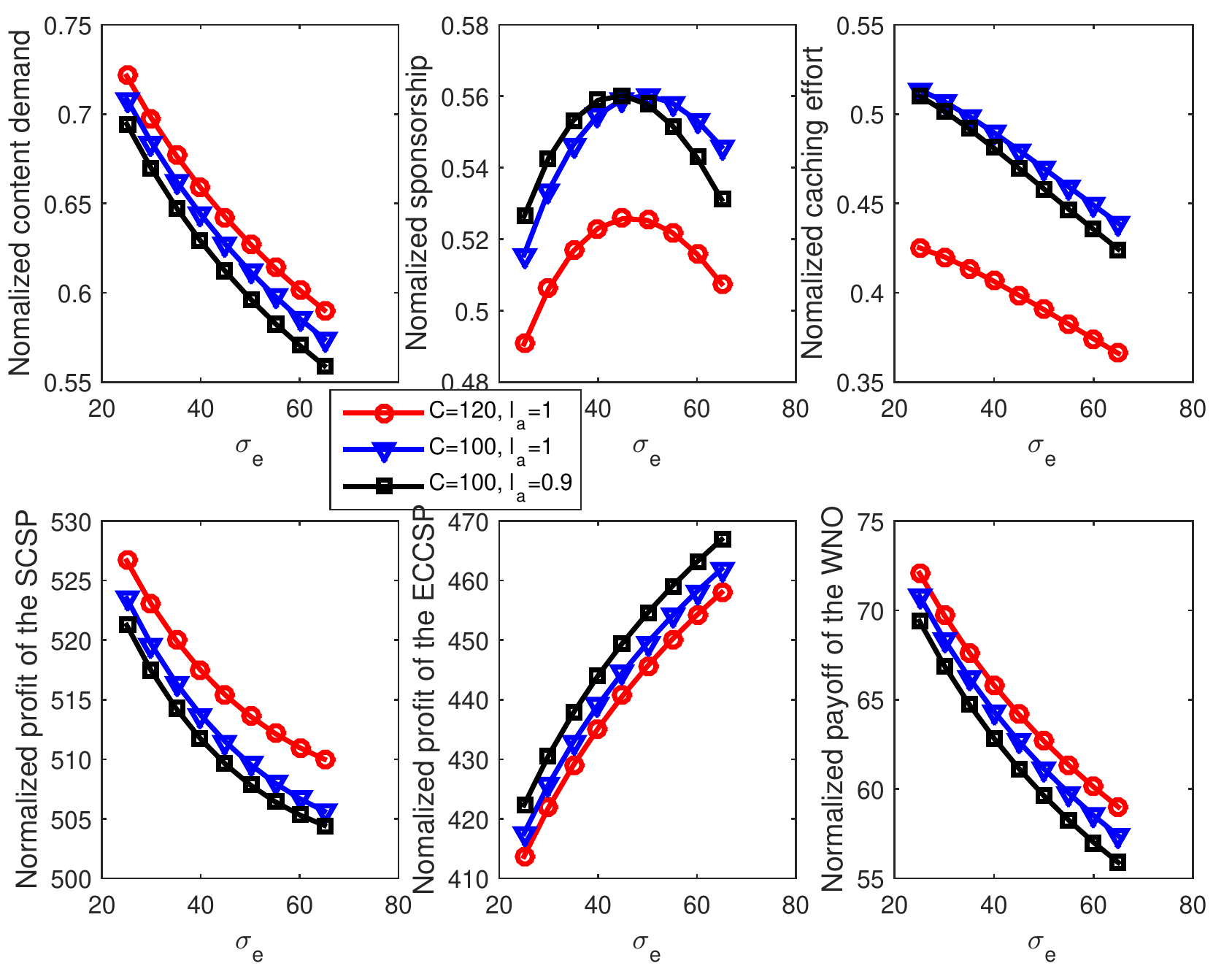}
	\caption{\small{The impact of the utility level of MUs on the players (MUs, SCSP, ECCSP, and WNO) from joint sponsored and edge caching content service market.}}\label{Fig:Performance_cegma_e}
	\end{figure}

We next study the impact of the utility coefficient of MUs on the players in the game model, as illustrated in Fig.~\ref{Fig:Performance_cegma_e}. We find that the sponsored content demand of the MU decreases with the increase of $\sigma_e$. This is because when $\sigma_e$ increases, the sponsorship fee from the SCSP becomes relative lower compared with the improved utility of the MU from consuming the content. Recall that the caching effort positively affects the utility of the MU from consuming the content, therefore, the sponsored content demand of the MU decreases. As a result, the SCSP wants to offer more sponsorship fee to attract the MU to consume the sponsored content, and the ECCSP has an incentive to lower its caching effort to reduce the caching cost. Once the sponsorship fee from the SCSP is large enough, the SCSP is not willing to offer more sponsorship fee to save its cost. Accordingly, the sponsorship fee from the SCSP increases first and then decreases. Therefore, the SCSP's profit decreases and the ECCSP's profit increases, which is consistent with the results in Fig.~\ref{Fig:Performance_cegma_e}. Since the sponsored content demand of the MU decreases, and thus the payment from the MU to the WNO is reduced, which leads to the decrease of the WNO's payoff. As expected, we find that the sponsored content demand of the MU increases and the caching effort of the ECCSP decreases as the content caching cost $C$ increases. This is due to the fact that, with the increase of $C$, the cost of the ECCSP for increasing caching effort becomes greater. In this case, the ECCSP is willing to reduce its caching effort for saving cost and thus compensate for its decreasing profit. Meanwhile, the SCSP wants to offer higher sponsorship fee to the MU for encouraging the MU's higher sponsored content demand. Consequently, both the SCSP's profit and the WNO's payoff increase. Moreover, from Fig.~\ref{Fig:Performance_cegma_e}, we find that the increase of the advertisement length $l_a$ leads to the increase of the SCSP's profit and the WNO's payoff. This is because when $l_a$ increases, the utility of the MU from consuming the content decreases, and the offered sponsorship fee becomes significant for the MU. This encourages higher sponsored content demand of the MU, which benefits the SCSP and the WNO accordingly. Therefore, the profit of the ECCSP decreases because of its decreasing content traffic.

	\begin{figure}[t]
	\centering
	\includegraphics[width=.45\textwidth]{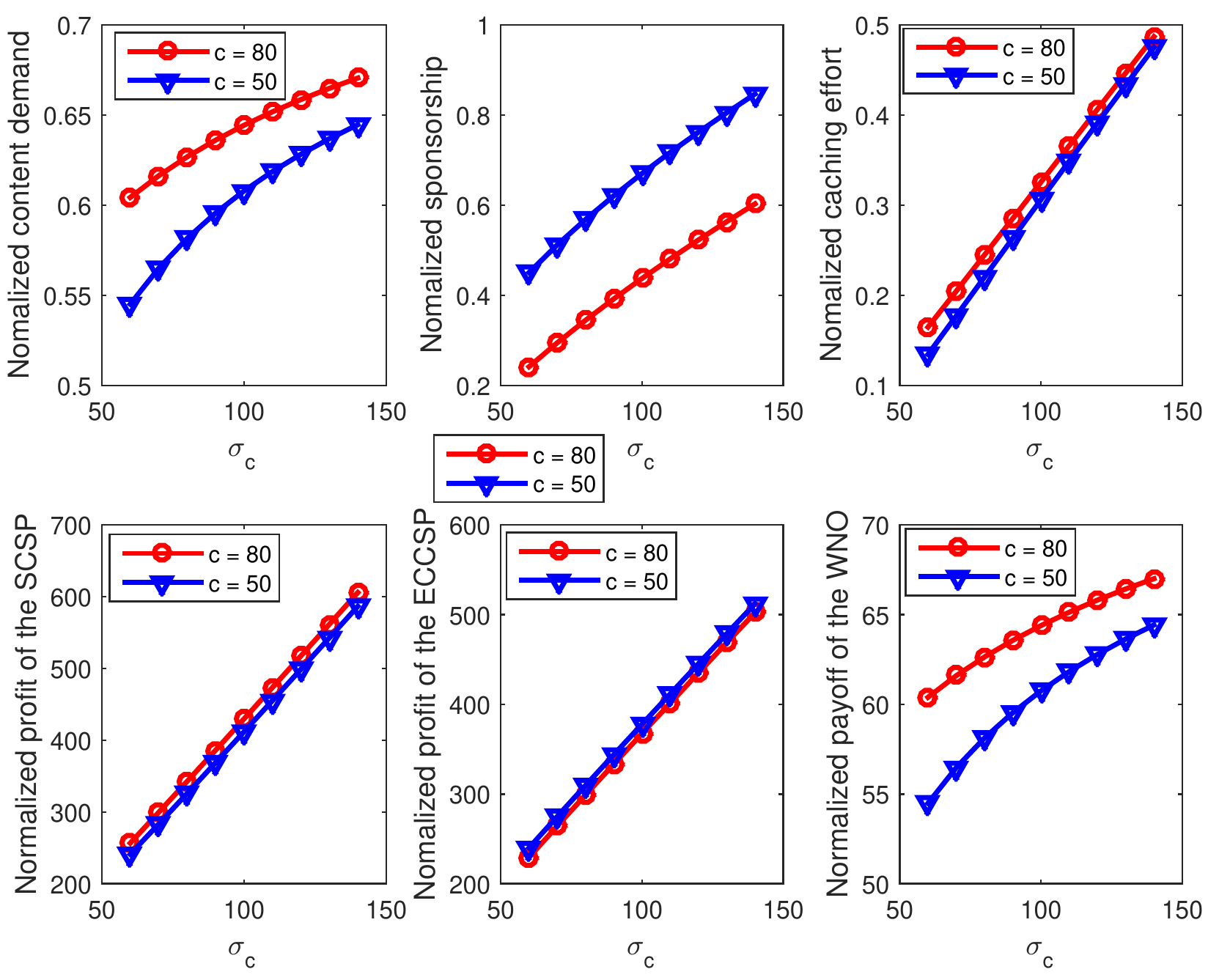}
	\caption{\small{The impact of the advertisement revenue coefficient on the players (MUs, SCSP, ECCSP, and WNO) from joint sponsored and edge caching content service market.}}\label{Fig:Performance_cegma_c}
	\end{figure}

Lastly, we examine the impact of the advertisement revenue coefficient $\sigma_c$ on the players in the game model, as illustrated in~Fig.~\ref{Fig:Performance_cegma_c}. We find that the profit of the SCSP, the profit of the ECCSP, and payoff of the WNO increases with the increase of $\sigma_c$. This is because the advertisement revenue from increasing content traffic is improved as $\sigma_c$ increases. Thus, the profit of both the SCSP and ECCSP are improved. In this case, the SCSP and ECCSP are more willing to increase the sponsorship fee and caching effort, respectively, to attract more demand from the MU. Recall that the optimal price is equal to the price constraint. As the price is fixed, with the increase of $\sigma_c$, the SCSP has an incentive to offer a higher sponsorship fee to the MU since its cost becomes insignificant compared with the increasing advertisement revenue coefficient. Thus, the sponsored content demand of the MU increases, which leads to the increase of the WNO's payoff consequently. By comparing curves with different value of handover cost $c$, we find that the decrease of $c$ leads to the decrease of the sponsored content demand, the SCSP's profit and the WNO's payoff. This is due to the fact that the MU is reluctant to choose the cached content from the ECCSP when the handover cost is high. Consequently, the MU is more willing to access and consume the content from the SCSP using. Due to the increasing content traffic, the WNO's payoff is improved accordingly. Utilizing this fact, the SCSP also has an incentive to reduce its sponsorship fee for saving the cost and achieving higher profit. However, the ECCSP needs to increase its caching effort to compensate for the higher handover cost. Therefore, the profit of the ECCSP decreases as shown in~Fig.~\ref{Fig:Performance_cegma_c}.

\section{Conclusion}\label{Sec:Conclusion}
In this work, we have considered a three-stage Stackelberg game based market model for a joint sponsored and edge caching content service. The model describes the interactions among the wireless operator, the sponsored content provider, and the edge provider. The analytical form of the sub-game perfect equilibrium in each stage has been analyzed by backward induction. We have also validated the existence of the Stackelberg equilibrium by capitalizing on the bilevel optimization technique. Furthermore, a sub-gradient based iterative algorithm has been proposed, which converges to the Stackelberg equilibrium with certainty. Numerical results have been proposed to evaluate the performance of all the game participants. Further studies of this work will focus on the social interactions among the users such as~\cite{nie2017socially}, or the dynamic pricing scheme such as~\cite{xiong2017network} in the context of sponsored content.

\section*{Acknowledgement}
This work was supported in part by WASP/NTU M4082187 (4080), Singapore MOE Tier 1 under Grant 2017-T1-002-007 RG122/17, MOE Tier 2 under Grant MOE2014-T2-2-015 ARC4/15, NRF2015-NRF-ISF001-2277, EMA Energy Resilience under Grant NRF2017EWT-EP003-041, ISF-NRF grant 2277/16, and in part by National Natural Science Foundation of China (No. 61601336).

\bibliography{bibfile}
\end{document}